%% file: main.tex
\documentclass[a4paper]{article}
\usepackage{latexsym,amsthm,amsmath,amssymb}

\theoremstyle{plain}
\newtheorem{theorem}{Theorem}
\newtheorem*{theorem*}{Theorem}
\newtheorem{lemma}[theorem]{Lemma}
\newtheorem*{lemma*}{Lemma}
\newtheorem{corollary}[theorem]{Corollary}
\newtheorem{observation}[theorem]{Observation}
\newtheorem*{observation*}{Observation}

\newtheorem*{conjecture*}{Conjecture}

\theoremstyle{definition}
\newtheorem{definition}[theorem]{Definition}

\usepackage{graphicx}
\usepackage[utf8]{inputenc}

\usepackage{authblk}
\usepackage[english]{babel}
\usepackage{graphicx}
\usepackage{type1ec}
\usepackage[T1]{fontenc}
\usepackage{bbm}
\usepackage{mathtools}
\usepackage{amsfonts}
\usepackage{dsfont}
\usepackage{pifont}
\usepackage{url}
\usepackage{multirow}
\usepackage{multicol}
\usepackage{pbox}
\usepackage{fancybox}
\usepackage{colortbl}
\usepackage{soul}
\usepackage{color}
\usepackage{tikz}
\usepackage{tkz-berge}
\usepackage{rotating}
\usepackage{pdflscape}
\usepackage{afterpage}
\usepackage{comment}
\usepackage{etoolbox}
\usepackage{environ}
\usepackage{dirtytalk}
\usepackage[nocompress]{cite}

\usepackage{algorithm}
\usepackage{algorithmic}
\usepackage{hyperref}

\usetikzlibrary{decorations,arrows,shapes,positioning}

\newcommand{\td}[1]{\mathbb{#1}}

\usepackage{complexity}

\newcommand{\bigO}[1]{{\mathcal{O}\!\left(#1\right)}}

\newcommand{\pname}[1]{\textsc{#1}}

\newcommand{\probl}[3]{
  \begin{flushleft}
    \fbox{
      \begin{minipage}{0.95\linewidth}
        \noindent {\pname{#1}}\\
        {\bf Instance:} #2\\
        {\bf Question:} #3
      \end{minipage}}
    \medskip
  \end{flushleft}
}

\newcommand{\inners}{1.2pt}
\newcommand{\outers}{1pt}

\newclass{\Hard}{hard}
\newclass{\Hness}{hardness}
\newcommand{\NPH}{\NP\text{-}\Hard}
\newcommand{\WH}{\W\textsc{[1]-}\Hard}

\newclass{\para}{para}
\newclass{\Complete}{complete}
\newclass{\Cness}{completeness}
\newcommand{\NPc}{\NP\textsc{-}\Complete}

\newfunc{\dist}{dist}
\newfunc{\girth}{girth}
\newfunc{\nd}{nd}
\newfunc{\YES}{YES}
\newfunc{\NOi}{NO}
\newfunc{\ff}{ff}
\newfunc{\dc}{dc}
\newfunc{\dcc}{d\overline{c}}
\newfunc{\ml}{ml}
\newfunc{\cf}{cf}
\newfunc{\tw}{tw}
\newfunc{\ext}{ext}
\newfunc{\pat}{pat}
\newfunc{\parts}{parts}
\newfunc{\clique}{clique}
\newcommand{\angled}[1]{\left\langle{#1}\right\rangle}

\newtoggle{proof_toggle}
\toggletrue{proof_toggle}
\NewEnviron{tproof} {
    \iftoggle{proof_toggle} {%
        \begin{proof}\BODY\end{proof}}
    {%
    }
}
\title{On structural parameterizations of the selective coloring problem}
\date{}

\author[1]{Guilherme C. M. Gomes}
\author[1]{Vinicius F. dos Santos}
\affil[1]{Departamento de Ci\^encia da Computa\c{c}\~{a}o, Universidade Federal de Minas Gerais -- Belo Horizonte, Brazil}
\begin{document}

\maketitle
%
%\authorrunning{G. C. M. Gomes et al.}
%\institute{Universidade Federal de Minas Gerais, Belo Horizonte, Minas Gerais, Brazil
%\email\{gcm.gomes, matheusresende, viniciussantos\}@dcc.ufmg.br\\}

    \begin{abstract}
    In the \pname{Selective Coloring} problem, we are given an integer $k$, a graph $G$, and a partition of $V(G)$ into $p$ parts, and the goal is to decide whether or not we can pick exactly one vertex of each part and obtain a $k$-colorable induced subgraph of $G$.
    This generalization of \pname{Vertex Coloring} has only recently begun to be studied by Demange et al. [Theoretical Computer Science, 2014], motivated by scheduling problems on distributed systems, with Guo et al. [TAMC, 2020] discussing the first results on the parameterized complexity of the problem.
    In this work, we study multiple structural parameterizations for \pname{Selective Coloring}.
    We begin by revisiting the many hardness results of Demange et al. and show how they may be used to provide intractability proofs for widely used parameters such as pathwidth, distance to co-cluster, and max leaf number.
    Afterwards, we present fixed-parameter tractability algorithms when parameterizing by distance to cluster, or under the joint parameterizations treewidth and number of parts, and co-treewidth and number of parts.
    Our main contribution is a proof that, for every fixed $k \geq 1$, \pname{Selective Coloring} does not admit a polynomial kernel when jointly parameterized by the vertex cover number and the number of parts, which implies that \pname{Multicolored Independent Set} does not admit a polynomial kernel under the same parameterization.
    \end{abstract}
    
    \input{intro.tex}
    \input{preliminaries.tex}
    \input{cluster.tex}
    \input{tw.tex}
    \input{cotw.tex}
    \input{vc.tex}
    \input{conclusion.tex}
    \bibliographystyle{abbrv}
    \bibliography{main}
    %\newpage
    %\input{sections/appendix.tex}
\end{document}

%% file: intro.tex
\section{Introduction}

Practical scheduling or task assignment problems are often modeled as graph coloring problems~\cite{mutual_exclusion_scheduling,scheduling_traffic,domain_decomposition}, with different constraints on the intended application leading to problems with significantly different complexities.
For instance, \pname{Vertex Coloring} can be solved in polynomial time on perfect graphs~\cite{perfect_color}, while \pname{Equitable Coloring} is \NPc\ on block graphs~\cite{equitable_dmtcs}.
In this work, we investigate a problem first discussed by Li and Simha~\cite{selective_wavelength} in the context of the wavelength division multiplexing optical networks, which they dubbed the \pname{Partition Coloring} problem, which is also known as \pname{Selective Coloring}; they showed that the problem was \NPH\ and began working on heuristics for it, which have since then been further developed by Noronha and Ribeiro~\cite{noronha}.
We adopt the later nomenclature, and formally define the problem as follows:

\probl{Selective Coloring}{A graph $G$, a partition $\mathcal{V}$ of $V(G)$, and an integer $k$.}{Is there a $k$-colorable induced subgraph of $G$ containing exactly one vertex of each part of $\mathcal{V}$?}

When the number of colors $k$ is fixed, we refer to the problem as \pname{Selective $k$-Coloring}.
In the complexity front, Demange et al.~\cite{selective_complexity} conducted an extensive study on the complexity of \pname{Selective $k$-Coloring} and the optimization version of \pname{Selective Coloring}, which they named \pname{Sel-Col}, on graph classes.
We only work with the decision version of the problem, so the following results are translations of their proofs regarding \pname{Sel-Col}.
On the negative side, Demange et al.~\cite{selective_complexity} proved that \pname{Selective Coloring} is \NPc\ on split graphs, complete multipartite graphs, and planar graphs of maximum degree three, that \pname{Selective $k$-coloring} is \NPc\ when $k = 1$ on paths, cycles, disjoint union of $C_4$'s, and on subcubic planar graphs.
Meanwhile, they showed that \pname{Selective Coloring} can be solved in polynomial time on disjoint union cliques, threshold graphs, graphs with stability number two, and, for every fixed $q$, on complete $q$-partite graphs.
Later on, Demange et al.~\cite{min_max_selective} settled the complexity of \pname{Selective $k$-coloring} for other graphs classes, showing that, when $k=1$, the problem is \NPc\ on subcubic planar unit disk graphs and permutation graphs and, when $k=2$, it is \NPc\ on twin graphs~\cite{twin_graph}.
They also worked on a worst case version of \pname{Selective Coloring}, i.e. the task to finding an induced subgraph of $G$ that contains one vertex of each part of $\mathcal{V}$ that has the maximum possible chromatic number.
In both papers, the authors further refined their analysis by imposing constraints on the parts of the partition, but we omit these discussions for brevity.

In terms of parameterized complexity, the results of Demange et al.~\cite{selective_complexity,min_max_selective} imply that \pname{Selective Coloring} is \para\NPH\ on multiple structural parameters, namely: treedepth, distance to disjoint paths, cotreewidth, max leaf number, distance to co-cluster, distance to bipartite, and feedback edge set, even when the number of colors is also used as parameter.
More recently, Guo et al.~\cite{selective_tamc} showed an \XP\ algorithm when parameterized by the number of parts in $\mathcal{V}$ and some initial parameterized complexity results.
At this point, it is important to note that, while \pname{Selective Coloring} is a clear generalization of \pname{Vertex Coloring}, \pname{Selective $1$-coloring} is equivalent to a central problem in parameterized complexity known as \pname{Multicolored Independent Set}~\cite{cygan_parameterized} where, given a graph $G$ and a $k$-coloring of its vertices, we are tasked with finding an independent set containing one vertex of each color.
Theorem 5.6 of Demange et al.~\cite{selective_complexity} implies that \pname{Multicolored Independent} is \para\NPH\ when parameterized by max leaf number if all parts have two or three vertices; it is currently unknown whether the latter condition is necessary to achieve hardness under this parameterization.

\smallskip
\noindent\textbf{Our results.} The main contributions of this work are complexity results for structural parameterizations of \pname{Selective Coloring}.
In particular, we show that \pname{Selective Coloring} is fixed parameter tractable when parameterized by distance to cluster, treewidth and number of parts of $\mathcal{V}$, and by cotreewidth and number of colors.
The first result generalizes the proof of Demange et al.~\cite{selective_complexity} that \pname{Selective Coloring} is polynomial time solvable on disjoint union of cliques, while the latter two imply that their reductions for graphs of constant treewidth and co-treewidth cannot be strengthened to also fix the number of parts and colors, respectively.
On the negative side, we show that, for every fixed $k \geq 1$, \pname{Selective $k$-Coloring} does not admit a polynomial kernel when simultaneously parameterized by vertex cover and number of parts unless $\NP \subseteq \coNP/\poly$, which implies that \pname{Multicolored Independent Set} has no polynomial kernel under the same parameterization and complexity hypothesis, which we believe to be of special interest to the community.

%% file: preliminaries.tex
\section{Notation and Terminology}

We refer the reader to~\cite{cygan_parameterized} for basic background on parameterized complexity, and recall here only some basic definitions.
A \emph{parameterized problem} is a language $L \subseteq \Sigma^* \times \mathbb{N}$. 
For an instance $I=(x,q) \in \Sigma^* \times \mathbb{N}$, $q$ is called the \emph{parameter}. 
A parameterized problem is \emph{fixed-parameter tractable} (\FPT) if there exists an algorithm $\mathcal{A}$, a computable function $f$, and a constant $c$ such that given an instance $(x,q)$, $\mathcal{A}$ correctly decides whether $I \in L$ in time bounded by $f(q) \cdot |I|^c$; in this case, $\mathcal{A}$ is called an \emph{\FPT\ algorithm}.
A kernelization	algorithm, or just \emph{kernel}, for a parameterized problem $\Pi$ takes an instance~$(x,q)$ of the problem and, in time polynomial in $|x| + q$, outputs an instance~$(x',q')$ such that $|x'|, q' \leqslant g(q)$ for some function~$g$, and $(x,q) \in \Pi$ if and only if $(x',q') \in \Pi$.
Function~$g$ is called the \emph{size} of the kernel and may be viewed as a measure of the ``compressibility'' of a problem using polynomial-time pre-processing rules.
A kernel is called \emph{polynomial} if $g(q)$ is a polynomial function in $q$.
A breakthrough result of Bodlaender et al.~\cite{distillation} gave the first framework for proving that some parameterized problems do not admit polynomial kernels, by establishing so-called \emph{composition algorithms}.
Together with a result of Fortnow and Santhanam~\cite{fortnow_santh}, this allows to exclude polynomial kernels under the assumption that $\NP \nsubseteq \coNP/\poly$, otherwise implying	a collapse of the polynomial hierarchy to its third level~\cite{uniform_non_uniform}; see~\cite{book_kernels} for a recent book on kernelization.

We use standard graph theory notation and nomenclature for our parameters, following classical textbooks in the areas~\cite{murty,cygan_parameterized}.
Define $[k] = \{1,\dots, k\}$.
A \emph{$k$-coloring} $\varphi$ of a graph $G$ is a function $\varphi: V(G) \mapsto~[k]$.
Alternatively, a $k$-coloring is a $k$-partition $V(G) \sim \{\varphi_1, \dots, \varphi_k\}$ such that $\varphi_i = \{u \in V(G) \mid \varphi(u) = i\}$.
Unless stated, all colorings are proper.
If $\mathcal{V}$ is a partition of $V(G)$ into $p$ parts and $S \subseteq [p]$, we say that $X \subseteq V(G)$ is \emph{$S$-selective} if, for every $i \in S$, $|X \cap \mathcal{V}_i| = 1$ and $X$ has no more vertices; we say that $X$ \textit{hits} $\mathcal{V}_i$ if $X \cap \mathcal{V}_i \neq \emptyset$.
A graph is a \textit{cluster graph} if each of its connected components is a clique; the \textit{distance to cluster} of a graph $G$, denoted by $\dc(G)$, is the size of the smallest set $U \subseteq V(G)$ such that $G - U$ is a cluster (co-cluster) graph.
Using the terminology of~\cite{cai_split}, a set $U \subseteq V(G)$ is an $\mathcal{F}$-\textit{modulator} of $G$ if the graph $G - U$ belongs to the graph class $\mathcal{F}$.
When the context is clear, we omit the qualifier $\mathcal{F}$.
For cluster graphs, one can decide if $G$ admits a modulator of size $k$ in time \FPT\ on $k$~\cite{clusterFPT}.

%% file: cluster.tex
\section{Selective Coloring parameterized by distance to cluster}
\label{sec:dc}

Our first goal is to prove that \pname{Selective Coloring} can be solved in \FPT\ time when parameterized by the distance to cluster of the input graph.
Throughout this section, we denote the modulator by $U$, the connected components of $G - U$ by $\mathcal{C} = \{C_1, \dots, C_r\}$, and by $\mathcal{V}(X)$ the parts of $\mathcal{V}$ that contain some vertex of $X \subseteq V(G)$.

The initial step of our algorithm is to first guess which of the $2^{|U|}$ subsets of $U$ shall be present in the solution and, afterwards, guess one of the $|U|^{|U|}$ possible colorings of this subset; the final step is to show how one can determine if these guesses can be extended to account for the vertices of $G - U$.
As such, suppose we are given $U$, $\mathcal{C}$, a subset $X \subseteq U$ that contains at most one vertex of each part of the $p$ parts of $\mathcal{V}$, and a coloring $\varphi'$ of $X$.
We build an auxiliary graph $H$ as follows:
$V(H) = \{s, t\} \cup A \cup W \cup (V(G) \setminus U) \cup P$, where $A = \{a_1, \dots, a_k\}$ represents the colors we may assign to vertices, $W = \{w_{ij} \mid i \in [k], j \in [r]\}$ whose role is to maintain the property of the coloring, $s$ is the source of the flow, $t$ is the sink of the flow, $V(G) = \{v_1, \dots, v_n\}$ are the vertices of $G$, and $P = \{\rho_1, \dots, \rho_p\}$ control that only one vertex may be picked per part of $\mathcal{V}$.
For the arcs, we have $E(H) = S \cup F \cup R \cup L \cup T$, where $S = \{(s, a_i) \mid i \in [k]\}$, $F = \{(a_i, w_{ij}) \mid i \in [k], j \in [r]\}$, $R = \{(w_{ij}, v_\ell) \mid v_\ell \in C_j, N(v_\ell) \cap \varphi'_i = \emptyset\}$, $L = \{(v_\ell, \rho_j) \mid v_\ell \in V_j\}$, and $T = \{(\rho_j, t) \mid V_j \cap X = \emptyset\}$.
As to the capacity of the arcs, we define $c : E(H) \rightarrow \mathbb{N}$, with $c(e \in S) = p$, and $c(e \in F \cup R \cup L \cup T) = 1$.
Semantically, the vertices of $A$ correspond to the $k$ colors, while each $w_{ij}$ ensures that cluster $C_j$ has at most one vertex of color $i$.
Regarding the arcs, $R$ encodes the adjacency between vertices of the clusters and colored vertices in $X$, and $L$ encodes which parts have already been hit by $X$.
Note that the arcs in $R$ and $L$ are the only ones affected by the pre-coloring $\varphi'$.
An example of the constructed graph can be found in Figure~\ref{fig:cluster_ex}.

%{\huge Apagar os vertices de $U$ do grafo auxiliar.}

\begin{figure}[!htb]
    \centering
        \begin{tikzpicture}[scale=0.7]
            %\draw[help lines] (-5,-5) grid (5,5);
            \GraphInit[unit=3,vstyle=Normal]
            \SetVertexNormal[Shape=circle, FillColor=white, MinSize=3pt]
            \tikzset{VertexStyle/.append style = {inner sep = \inners, outer sep = \outers}}
            \SetVertexLabelOut
            \begin{scope}
                \begin{scope}
                    \Vertex[x=0, y=0, Lpos=180,Math]{v_1}
                    \Vertex[x=2, y=0, Lpos=0, Math]{v_2}
                    \Vertex[x=1, y=0, Lpos=270, Math]{v_7}
                    \Vertex[x=0, y=2, Lpos=90,Math]{v_3}
                    \Vertex[x=2, y=2, Lpos=90,Math]{v_4}
                    \Vertex[x=0, y=-2, Lpos=-90,Math]{v_5}
                    \Vertex[x=2, y=-2, Lpos=-90,Math]{v_6}
                    \Edges(v_1,v_7,v_2,v_4,v_3,v_1,v_5,v_6,v_2)
                \end{scope}
                \draw (-1, -0.7) rectangle (3, 0.5);
                \Edge[style={bend left}](v_1)(v_2)
                \begin{scope}
                    \tikzset{VertexStyle/.append style = {shape = diamond, inner sep = 2pt}}
                    %\Vertex[Node]{c0}
                    \AddVertexColor{black}{v_3}
                    \AddVertexColor{white}{v_2, v_4}
                \end{scope}
                \begin{scope}
                    \tikzset{VertexStyle/.append style = {shape = rectangle, inner sep = 3pt}}
                    %\Vertex[Node]{c1}
                    \AddVertexColor{black}{v_6}
                    \AddVertexColor{white}{v_1,v_5, v_7}
                \end{scope}
            \end{scope}
        \end{tikzpicture}
        \hfill
        \begin{tikzpicture}[scale=0.7]
            \GraphInit[unit=3,vstyle=Normal]
            \SetVertexNormal[Shape=circle, FillColor=white, MinSize=3pt]
            \tikzset{VertexStyle/.append style = {inner sep = \inners, outer sep = \outers}}
            \SetVertexLabelOut
            \begin{scope}
                \Vertex[x=0,y=0,Lpos=180, Math]{s}
                \Vertex[x=1, y=1, Lpos=90, Math]{a_1}
                \Vertex[x=1, y=-1, Lpos=-90, Math]{a_2}
                \Vertex[x=9, y=0, Lpos=0, Ldist=1pt, Math]{t}
                
                \foreach \i in {1,2} {
                    \Edge[style={->, dashed}](s)(a_\i)
                    \foreach \j in {1,2} {
                        \pgfmathsetmacro{\y}{-3*(\i-1) + 1.5 - (1*(\j-1) - 0.5)}
                        \pgfmathtruncatemacro{\id}{2*(\i-1) + \j-1 + 3}
                        \pgfmathtruncatemacro{\lpos}{180*(\i-1) + 90}
                        \Vertex[x=3, y=\y, Lpos=\lpos, Math, L={w_{\i\j}}]{w_\i\j}
                        \Vertex[x=5, y=\y, Lpos=\lpos, Math, L ={v_\id}]{cv_\id}
                        \Edge[style={->}](a_\i)(w_\i\j)
                        %\Edge[style={->}](cv_\id)(t)
                    }
                }

                \begin{scope}
                    \tikzset{VertexStyle/.append style = {shape = rectangle, inner sep = 2.2pt}}
                    \Vertex[x=7, y=1.5, Lpos=90, Math, L={\rho_1}]{r1}
                    \Vertex[x=7, y=-1.5, Lpos=270, Math, L={\rho_4}]{r4}
                    \AddVertexColor{black}{r4}
                \end{scope}
                \begin{scope}
                    \tikzset{VertexStyle/.append style = {shape = diamond, inner sep = 2.2pt}}
                    \Vertex[x=7, y=0.5, Lpos=90, Math, L={\rho_2}, Ldist=-1pt]{r2}
                    \Vertex[x=7, y=-0.5, Lpos=270, Math, L={\rho_3}, Ldist=-1pt]{r3}
                    \AddVertexColor{black}{r2}
                \end{scope}
                
                \Edge[style={->, bend left=20}](w_21)(cv_3)
                \Edge[style={->}](w_11)(cv_4)
                \Edge[style={->}](w_22)(cv_5)
                \Edge[style={->, bend right=20}](w_12)(cv_6)
                
                \Edge[style={->}](cv_3)(r2)
                \Edge[style={->}](cv_4)(r3)
                \Edge[style={->, bend left=20}](cv_5)(r1)
                \Edge[style={->}](cv_6)(r4)
                
                \foreach \i in {2,4} {
                    \Edge[style={->}](r\i)(t)
                }

            \end{scope}
    \end{tikzpicture}
\caption{{(left)} The input graph with $U = \{v_1, v_2, v_7\}$ and parts defined by vertices of different shapes and colors; {(right)} Auxiliary graph constructed from the pre-coloring of $X = \{v_1, v_2\}$ that has $\varphi'(v_i) = i$. Solid arcs have unit capacity; dashed arcs have capacity equal to $|\mathcal{V}|$.\label{fig:cluster_ex}}
\end{figure}
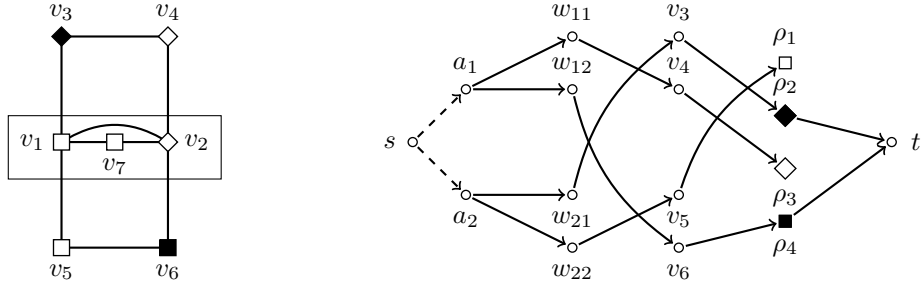

Now, let $f : E(H) \rightarrow \mathbb{N}$ be the function corresponding to the max-flow from $s$ to $t$ obtained using any of the algorithms available in the literature~\cite{maxflow}.
Our first observation, as given by the following lemma, is that, if no $(s,t)$-flow saturates the outbound arcs of $s$, then we cannot extend $\varphi'$.

\begin{lemma}\label{lem:partial_flow}
    If there is some $e \in T$ with $f(e) = 0$, then $G$ has no induced subgraph that hits every part of $\mathcal{V}$ once, contains $X$, and admits a $k$-coloring that extends $\varphi'$.
\end{lemma}

\begin{proof}
    By contraposition, suppose that there is a coloring $\varphi$ satisfying the conditions of the statement, and let $G^*$ be the induced subgraph of $G$ that contains $X$ and is $[p]$-selective.
    For each $v \in C_j \cap \varphi_i \cap \mathcal{V}_\ell$, we push one flow unit along the path $\angled{s, a_i, w_{ij}, v, \rho_\ell, t}$.
    This process does not exceed the capacities of the involved arcs because: (i) we use at most $p$ units of flow, so $f((s,a_i)) \leq p = c((s, a_i))$, (ii) $\varphi$ is a proper coloring of $G^*$, so at most one vertex of clique $C_j$ is in $\varphi_i$, implying $f((a_i, w_{ij})) \leq 1$, (iii) to each $v \in V(G^*) \cap C_j$ is assigned a single color $i$, so $f((w_{ij}, v)) \leq 1$, and (iv) $V(G^*)$ is $[p]$-selective, so there is a unique $v \in V(G^*) \cap \mathcal{V}_\ell$, implying $f((\rho_\ell, t)) = 1$, for every $(\rho_\ell, t) \in T$. 
\end{proof}

\begin{lemma}
    \label{lem:col_construction}
    The maximum ($s,t$)-flow given by $f$ is equal to $|\mathcal{V} \setminus \mathcal{V}(X)|$ if and only if there is an induced subgraph $G^*$ of $G$ that hits every part of $\mathcal{V}$ once, contains $X$, and admits a $k$-coloring that extends $\varphi'$.
\end{lemma}

 \begin{proof}
    For the forward direction, we construct $G^*$ starting with $X$ and then picking, for each $\mathcal{V}_\ell \in \mathcal{V} \setminus \mathcal{V}(X)$, the unique vertex $v \in V(H) \cap V(G)$ that has $f((v, \rho_\ell)) = 1$; such a vertex must exist for each $\mathcal{V}_\ell \notin \mathcal{V}(X)$ since the maximum flow is equal to $|\mathcal{V} \setminus \mathcal{V}(X)|$ and through each $\rho_\ell$ passes a different flow unit.
    As to the coloring $\varphi$, we set $\varphi(u \in X) = \varphi'(u)$ and, for each $j \in [r]$ and $v \in V(G^*) \cap C_j$, set $\varphi(v) = i$ if and only if $f((w_{ij}, v)) = 1$.
    This coloring is proper since $\varphi'$ is proper and for every arc $(a_i, w_{ij}) \in E(H)$, we have that $f((a_i, w_{ij})) \leq 1$, so no two vertices of a clique have the same color, and arc $(w_{ij}, v)$ is in $E(H)$ if and only if $v$ has no neighbor in $X$ colored with $i$.
    As such, we have that $V(G^*)$ is $[p]$-selective and is $k$-colorable.
    For the converse, it suffices to apply the contraposition of Lemma~\ref{lem:partial_flow}.
 \end{proof}

At this point we are essentially done.
Lemmas~\ref{lem:partial_flow} and~\ref{lem:col_construction} guarantee that, if the max-flow algorithm fails to yield a large enough flow, the fixed pre-coloring $\varphi'$ of $X \subseteq U$ cannot be extended; moreover, the latter also implies that, if an extension is possible, max-flow correctly finds it.

\begin{theorem}
    \label{thm:dc_fpt}
    \pname{Selective Coloring} parameterized by distance to cluster can be solved in \FPT\ time.
\end{theorem}

\begin{proof}
    For each $X \subseteq U$ and each a coloring $\varphi'$ of $X$, we construct $H$ as before and execute a max-flow algorithm.
    By Lemma~\ref{lem:col_construction}, the max flow of $H$ is equal to $|\mathcal{V} \setminus \mathcal{V}(X)|$ if and only if there is an induced subgraph $G^*$ of $G$ that hits every part of $\mathcal{V}$ exactly once, contains $X$, and admits a $k$-coloring that extends $\varphi'$.
    Guess $X$, guess $\varphi'$, summon construction, profit.
    If no choice of $(X, \varphi')$ admits a valid extension, $(G, \mathcal{V}, k)$ is a negative instance of \pname{Selective Coloring}.
    Otherwise, there is some $(X, \varphi')$ that can be extended and, again by Lemma~\ref{lem:col_construction}, we can find the corresponding coloring in polynomial time.
    As to the complexity, there are at most $2^{|U|}$ choices of $X$ and, for each $X$, $B_{|X|}$ possible colorings of $X$, where $B_n$ is the $n$-th Bell number; consequently, our algorithm runs in $\bigO{2^{|U|}B_{|U|}\poly(n)}$ time if $|V(G)| = n$.
\end{proof}
    
%   It is worthy to note here that there is nothing special about the capacities of the arcs in $S$; they act only as upper bounds to the number of vertices a color may be assigned to.
%   Thus, not surprisingly, the same algorithm applies to problems where the size of each color class is only upper bounded.
%   This will be particularly useful in the next session.
%   Looking at the proof of Theorem~\ref{thm:dc_fpt}, the only non-polynomial step is guessing the coloring of the modulator.
%   A straightforward corollary is that if there is a polynomial number of distinct colorings of $U$ and this family can be computed in polynomial time, we can apply the same ideas and check if an equitable $k$-coloring of the input graph exists in polynomial time.
%   In particular, unipolar graphs satisfy the above condition.
%   If we parameterize by distance to unipolar the problem remains \WH\ due to the hardness for split graphs.
%   On the other hand, if we parameterize by distance to unipolar $d$ and the number of colors $k$ we have an \FPT\ algorithm: the central clique of $G - U$ has at most $k$ vertices, so we can treat $G$ as a graph with distance to cluster at most $k + d$ and apply Theorem~\ref{thm:dc_fpt}.
%   
%   \begin{corollary}
%       \pname{Equitable Coloring} on unipolar graphs is in \P.
%       When parameterized by distance to unipolar, the problem remains \WH; if also parameterized by the number of colors, there is an \FPT\ algorithm that solves it.
%   \end{corollary}

%% file: tw.tex
\section{Treewidth and number of parts}
\label{sec:tw}

As shown in~\cite{selective_complexity}, for $k = 1$, \pname{Selective $k$-Coloring} is \NPH\ on disjoint union of $P_3$'s, implying that the problem is \para\NPH\ when jointly parameterized by treedepth and distance to disjoint paths.
Their proof employed a reduction from \pname{3-SAT} and the size of $\mathcal{V}$ was linear in the number of variables of the formula.
We show that this reduction cannot be strengthened to also bound the number of parts by presenting an \FPT\ algorithm when parameterizing by the treewidth $t$ and number of parts $p$.
Before proceeding to the algorithm, note that we may assume that $k \leq t$, otherwise the input graph $G$ itself is $k$-colorable, so all of its induced subgraphs will be $k$-colorable, and $(G, \mathcal{V}, k)$ would be a trivially positive instance of \pname{Selective Coloring}.

A \textit{tree decomposition} of a graph $G$ is a pair $\td{T} = \left(T, \mathcal{B} = \{B_j \mid j \in V(T)\}\right)$, where $T$ is a tree and $\mathcal{B} \subseteq 2^{V(G)}$ is a family where: $\bigcup_{B_j \in \mathcal{B}} B_j = V(G)$;
for every edge $uv \in E(G)$ there is some~$B_j$ such that $\{u,v\} \subseteq B_j$;
for every $i,j,q \in V(T)$, if $q$ is in the path between $i$ and $j$ in $T$, then $B_i \cap B_j \subseteq B_q$.
Each $B_j \in \mathcal{B}$ is called a \emph{bag} of the tree decomposition.
$G$ has treewidth at most $t$ if it admits a tree decomposition such that no bag has more than $t+1$ vertices.
For more on treewidth, we refer to~\cite{treewidth}.
After rooting $T$, $G_x$ denotes the subgraph of $G$ induced by the vertices contained in any bag that belongs to the subtree of $T$ rooted at bag $x$.
An algorithmically useful property of tree decompositions is the existence of a \emph{nice tree decomposition} of minimum width and $\bigO{t|V(G)|}$ nodes.

\begin{definition}[Nice tree decomposition]
    A tree decomposition $\td{T}$ of $G$ is said to be \emph{nice} if its tree is rooted at, say, the empty bag $r(T)$ and each of its bags is from one of the following four types:
    \begin{enumerate}
        \item \emph{Leaf node}: a leaf $x$ of $\td{T}$ with $B_x = \emptyset$.
        \item \emph{Introduce vertex node}: an inner bag $x$ of $\td{T}$ with one child $y$ such that $B_x \setminus B_y = \{u\}$.
        \item \emph{Forget node}: an inner bag $x$ of $\td{T}$ with one child $y$ such that $B_y \setminus B_x = \{u\}$.
        \item \emph{Join node}: an inner bag $x$ of $\td{T}$ with two children $y,z$ such that $B_x = B_y = B_z$.
    \end{enumerate}
\end{definition}

\begin{theorem}
    \label{thm:tw_fpt}
    \pname{Selective Coloring} parameterized by treewidth $t$ and number of parts $p$ is in \FPT.
\end{theorem}

\begin{proof}
    We present a dynamic programming algorithm over a nice tree decomposition $\td{T} = (T, \mathcal{B})$ of $G$, which we assume to be given in the input; w.l.o.g. we assume that $T$ is rooted at a forget node.
    For each node $x$ of $T$, we have a table $f_x$ indexed by $S \in 2^{[p]}$ and a coloring $\varphi$ of the vertices in $B_x$.
    We want to show that $f_x(S, \varphi) = 1$ if and only if there is a $k$-colorable induced subgraph of $G_x$ where, for each $j \in S$, there is some vertex in $V(G_x) \cap \mathcal{V}_j \setminus B_x$, and for every $B_x$ can be colored according to $\varphi$.
    For the next cases, let $x$ be the node we are interested in solving; also, let us denote by $\varphi^{\uparrow v,i}$ the coloring obtained by extending $\varphi$ to include $v$ on color class $i$, and $\varphi^{\downarrow v}$ the restriction of $\varphi$ obtained by discoloring $v$.
    To simplify our transition functions, whenever there is some vertex $u$ colored by $\varphi$ but $\mathcal{V}(u) \in S$, we say that $f_x(S, \varphi) = 0$, since this would imply that two distinct vertices of the same part were picked by the solution;
    we also set $f_x(S, \varphi) = 0$ if some color class of $\varphi$ does not induce an independent set of $B_x$.
    Note that these checks can be performed in $\bigO{n^2}$ time.
    
    \emph{Leaf node.} Since $B_x = \emptyset$ and $x$ has no children in $T$, we set $f_x(\emptyset, \emptyset) = 1$.
    
    \emph{Introduce node.} Let $y \in V(T)$ be the child of $x$ and $v \in B_x \setminus B_y$.
    We compute $f_x(S, \varphi)$ according as follows:
    
    \begin{equation}
        \label{eq:tw_intro}
        f_x(S, \varphi) = 
        \begin{cases}
            f_y(S, \varphi), &\text{ if $\varphi^{\downarrow v} = \varphi$;}\\
            f_y(S, \varphi^{\downarrow v}), &\text{ otherwise.}\\
        \end{cases}
    \end{equation}
    
    In the first case of Equation~\ref{eq:tw_intro}, we are dealing with the situation where $v$ was not colored by $\varphi$, so any solution to $G_y$ is valid for $G_x$ under the same constraints imposed by $S$ and $\varphi$.
    Otherwise, $v$ was colored by $\varphi$ and, since $\mathcal{V}(v) \notin S$, any partial solution of $G_y$ that has not picked a vertex of the same part as $v$ and has no neighbor of $v$ colored with $\varphi(v)$ can be directly extended include $v$, under the constraints given by $S$.
    
    \emph{Forget node.} Again, let $y \in V(T)$ be the child of $x$ but $v \in B_y \setminus B_x$.
    We compute $f_x$ as follows:
    
    \begin{equation}
        \label{eq:tw_forget}
        f_x(S, \varphi) = 
        \begin{cases}
            f_y(S, \varphi), &\text{ if $\mathcal{V}(v) \notin S$;}\\
            \max\left\{f_y(S, \varphi), \max\limits_{i \in [k]}\left\{f_y(S \setminus \mathcal{V}(v), \varphi^{\uparrow v, i})\right\}\right\}, &\text{ otherwise.}\\
        \end{cases}
    \end{equation}
    
    For the first case, since $\mathcal{V}(v) \notin S$, $v$ cannot be in any solution of $G_x$ and, consequently, every solution to $G_x$ under $S$ and $\varphi$ must also be a solution to $G_y$ under the same $S$ and $\varphi$.
    For the second case, we know that some vertex of $\mathcal{V}(v)$ has been selected, which could be some vertex different from $v$, i.e. $v$ was not colored in $G_y$, or it was $v$, but in this case we must check, for each possible coloring of $v$, if it is possible to color find a solution to $G_y$ that uses no other vertex of $\mathcal{V}(v)$ and $B_y$ is colored according to an extension of $\varphi$ that includes $v$, i.e. there must be a solution represented by $f_y(S \setminus \mathcal{V}(v), \varphi^{\uparrow v, i})$, for some $i \in [k]$.
    
    \emph{Join node.} Let $y,z \in V(T)$ be the children of $x$.
    We transition according to Equation~\ref{eq:tw_join}.
    
    \begin{equation}
        \label{eq:tw_join}
        f_x(S, \varphi) = \max\limits_{R \subseteq S}\left\{f_y(R, \varphi)f_x(S \setminus R, \varphi)\right\}
    \end{equation}
    
    If $G_x$ is a solution respecting $S$ and $\varphi$, there are vertices in $G_y \setminus B_x$ that hit some of the parts required by $S$, while all other parts must be hit by vertices of $G_z \setminus B_x$; since $B_x$ separates $G_y \setminus B_x$ and $G_z \setminus B_x$, no part may be hit by partial solutions to both $G_y \setminus B_x$ and $G_z \setminus B_x$.
    Since, \textit{a priori}, we do not know how these hits are spread across $G_y \setminus B_x$ and $G_z \setminus B_x$, we must test, for each $R \subseteq S$, if there is a solution to $G_y$ restricted to $R$ and $\varphi$ and a solution to $G_z$ restricted to $S \setminus R$ and $\varphi$, which is what is computed by Equation~\ref{eq:tw_join}.
    
    Since $T$ is rooted at a forget node $r$, we may read the solution to the problem by checking if $f_r([p], \emptyset) = 1$, i.e. there is a solution to $G_r = G$ that touches all $p$ parts of $\mathcal{V}$.
    The correctness of the algorithm follows from our previous arguments.
    As to the running time, for each $x \in V(T)$, we have $\bigO{2^pB_t}$ entries to $f_x$, where $B_n$ is the $n$-th Bell number.
    For each entry, we need $\bigO{1}$ time for leaf and introduce nodes, $\bigO{t}$ time for forget nodes, and $\bigO{2^p}$ time for join nodes, totalling a running time of the order of $\bigO{4^pB_tn^3}$ time.
\end{proof}

%% file: cotw.tex
\section{Cotreewidth and number of colors}

Cotreewidth, the treewidth of the complementary graph, has recently begun to attract the attention of the community.
It has been shown to yield \FPT\ algorithms for problems which are \WH\ for either treewidth or cliquewidth, such as \pname{Equitable Coloring}~\cite{equitable_dmtcs} and dense subgraph detection~\cite{dense_cotreewidth}.
This is not the case, however, for \pname{Selective Coloring}: Demange et al.~\cite{selective_complexity} have shown that the problem is \NPH\ for complete $q$-partite graphs even when the size of each independent set is at most three, i.e. the complementary graph is the disjoint union of triangles, which has treewidth equal to two, and implies that \pname{Selective Coloring} is \para\NPH\ when parameterized by cotreewidth.
While we have shown in Section~\ref{sec:tw} that the problem is fixed parameter tractable under treewidth and number of parts, here we show that we can also attain tractability when parameterizing by cotreewidth and \textit{number of colors}.
As done with \pname{Equitable Coloring}~\cite{equitable_dmtcs}, we deal with the complementary problem parameterized by treewidth and number of colors, i.e. \pname{Selective Clique Partition}, which, given $G$, $\mathcal{V}$, and $k$, asks for a partition into cliques of size $k$ that uses one vertex of each part of $\mathcal{V}$.
We first show that the intuitively more powerful parameterization treewidth and number of parts is actually equivalent to treewidth and number of colors.

\begin{observation}
    \label{obs:sel_clique}
    If $p > kt$, then $(G, \mathcal{V})$ is a negative instance of \pname{Selective Clique Partition}.
\end{observation}

\begin{proof}
    Suppose to the contrary that there is a partition into cliques of $G$ given by $\{C_1, \dots, C_k\}$.
    Since each $C_i$ is a clique, it must be entirely contained in some bag of $\td{T}$, i.e. we have that $|C_i| \leq t$ and, thus, $\sum_{i \in [k]} |C_i| \leq kt$.
    Since our solution hits each part of $\mathcal{V}$ exactly once, it also holds that $\sum_{i \in [k]} |C_i| = p$, a contradiction.
\end{proof}

\begin{theorem}
    \label{thm:cotw}
    \pname{Selective Clique Partition} parameterized by treewidth and number of parts is in \FPT.
\end{theorem}

\begin{proof}
    Much like in Theorem~\ref{thm:tw_fpt}, our approach is to show a dynamic programming algorithm over the nice tree decomposition $\td{T} = (T, \mathcal{B})$, rooted at a forget node.
    For each node $x \in V(T)$, our table $f_x$ is indexed by the triple $(S, Q, \ell)$, where $S \subseteq [p]$, $Q \subseteq B_x$, and $\ell \leq k$; our goal is to show that $f_x(S, Q, \ell) = 1$ if and only if we can pick vertices of $G_x \setminus B_x$ that hit the parts given by $S$ once, $Q$ has no vertex in a part in $S$, and the picked vertices in $G_x \setminus B_x$, along with the vertices in $Q$, can be arranged in $\ell$ cliques.
    If $Q$ has a vertex in some part $\mathcal{V}_j$ where $j \in S$, $\ell < 0$, or $Q$ has two vertices of the same part of $\mathcal{V}$, we set $f_x(S, Q, \ell) = 0$.
    
    \emph{Leaf node.} Since $B_x = \emptyset$ and $x$ has no children, we set $f_x(\emptyset, \emptyset, 0) = 1$.

    \emph{Introduce node.} Let $y \in V(T)$ be the child of $x$ and $v \in B_x \setminus B_y$.
    We compute $f_x(S, Q, \ell)$ according to the following equation, where $\clique(R) = 1$ if and only if $R$ is a clique; note that if there is some $u \in R$ with $\mathcal{V}(u) \in S$, then $u \in Q$ and we would have set $f_x(S, Q, \ell) = 0$.
    
    \begin{equation}
        \label{eq:cotw_intro}
        f_x(S, Q, \ell) = 
        \begin{cases}
            f_y(S, Q, \ell), &\text{ if $v \notin Q$;}\\
            \max\limits_{\{v\} \subseteq R \subseteq Q} f_y(S, Q \setminus R, \ell - 1)\cdot\clique(R), &\text{ otherwise.}
        \end{cases}
    \end{equation}
    
    In the first case of Equation~\ref{eq:cotw_intro}, any solution to $G_x$ that does not use $v$, which is the unique vertex in $V(G_y) \setminus V(G_x)$ is simply a solution to $G_y$ which was trivially extended to $G_x$.
    For the second case, if $v$ must be part of the solution $G^*_x$ to $G_x$ and $\mathcal{V}(v) \notin S$, we must cover $v$ with some clique of $G_x$ but, since $x$ introduces $v$, this new clique $R$ must be formed by vertices of $Q$, i.e. they must not have been covered in $G_y$.
    That is, $G^*_x$ can be obtained by extending a solution of $G_y$ that covers $Q \setminus R$, $G_y \setminus B_y$, is $S$-selective, and uses $\ell - 1$ cliques to do so.

    \emph{Forget node.} Again, let $y \in V(T)$ be the child of $x$ but $v \in B_y \setminus B_x$.
    We compute the table for $x$ according to the function:
    
    \begin{equation}
        \label{eq:cotw_forget}
        f_x(S, Q, \ell) = 
        \begin{cases}
            f_y(S, Q, \ell), &\text{ if $\mathcal{V}(v) \notin S$;}\\
            \max\left\{f_y(S \setminus \mathcal{V}(v), Q \cup \{v\}, \ell), f_y(S, Q, \ell)\right\}, &\text{ otherwise.}
        \end{cases}
    \end{equation}
    
    If $\mathcal{V}(v) \notin S$, then $v$ cannot be chosen in a solution to $G_x$ that respects $(S, Q, \ell)$, so it follows that, to obtain such a solution for $G_x$, any solution to $G_y$ that also respects $(S, Q, \ell)$ will be a solution to $G_x$ under these constraints.
    Otherwise, some vertex of $\mathcal{V}(v)$ must be covered by the $\ell$ cliques of $G_x$; this vertex could have been $v$ itself, and any solution to $G_y$ respecting $(S \setminus \mathcal{V}(v), Q \cup \{v\}, \ell)$ can be extended to a solution to $G_x$ that respects $(S, Q, \ell)$, or this vertex is not $v$, but then this vertex is not in $B_x$ nor $B_y$, since $\{v\} = B_y \setminus B_x$ and $\mathcal{V}(u) \cap S = \emptyset$ for every $u \in Q$, so any solution to $G_y$ that respects $(S, Q, \ell)$ will also be a solution to $G_x$ under the same constraints.

    \emph{Join node.} For a join node $x$, let $y,z \in V(T)$ be the children of $x$ but $v \in B_y \setminus B_x$.
    
    \begin{equation}
        \label{eq:cotw_join}
        f_x(S, Q, \ell) = \max\limits_{S_y \subseteq S}\ \max\limits_{Q_y \subseteq Q}\ \max\limits_{\ell_y \leq \ell}\left\{f_y(S_y, R, \ell_y)\cdot f_z(S \setminus S_y, Q \setminus Q_y, \ell - \ell_y)\right\}
    \end{equation}
    
    Any solution to $G_x$ under $(S, Q, \ell)$ can be seen as the union of two partial solutions $G^*_y$ and $G^*_z$ to $G_y$ and $G_z$, respectively.
    In particular, since $V(G^*_y) \cap V(G^*_z) \subseteq V(G_y) \cap V(G_z) = B_x$, only a subset $S_y \subseteq S$ of the parts may have been hit by the $\ell_y$ cliques of $G^*_y$ using $Q_y \subseteq Q$ of the vertices of $B_y = B_x$, while the other $S \setminus S_y$ parts must have been hit by $\ell - \ell_y$ cliques picked by $G^*_z$ using the remaining $Q \setminus Q_y$ vertices of $B_z = B_x$.
    Since this choice of $(S_y, Q_y, \ell_y)$ is unknown at first, we must test all $\bigO{2^{|S|+|Q|}k}$ possibilities.
    
    To read the solution to $(G, \mathcal{V}, k)$ we look at the root node $r \in V(T)$.
    Since it is a forget node and has no parent, $B_r = \emptyset$; consequently, $G_r = G$ has a solution that touches all parts of $\mathcal{V}$ if and only if $F_r([p], \emptyset, k) = 1$.
    Since the most expensive nodes to be computed are the join nodes, it holds that our algorithm runs in $\bigO{3^{t+p}k^2t^2n}$, where the $t^2$ factor is obtained by computing the $\clique(R)$ function of introduce nodes.
\end{proof}

%The next corollaries follow immediately from Observation~\ref{obs:sel_clique}.

\begin{corollary}
    When parameterized by treewidth and number of colors, \pname{Selective Clique Partition} can be solved in \FPT\ time.
\end{corollary}

\begin{corollary}
    \pname{Selective Coloring} parameterized by cotreewidth and number of colors is in \FPT.
\end{corollary}

%% file: vc.tex
\section{Negative kernelization results}

While our previous result about the fixed parameter tractability of \pname{Selective Coloring} under distance to cluster implies the existence of an \FPT\ algorithm when parameterized by the vertex cover number, nothing was known in terms of kernelization for either parameter.
In this section we show that, for every fixed $k \geq 1$, \pname{Selective Coloring} does not admit a polynomial kernel when jointly parameterized by vertex cover and number of parts in the partition; throughout this section, we refer to this problem as \pname{Selective $k$-Coloring}
This implies that \pname{Multicolored Independent Set} does not admit a kernel when parameterized by vertex cover and number of color classes, contrasting with the polynomial kernel for \pname{Independent Set} parameterized by vertex cover~\cite{minor_free_kernel}.
We are going to show that the \NPc\ problem \pname{3-Coloring} on 4-regular graphs~\cite{quartic_coloring} OR-cross-composes~\cite{cross_composition} into \pname{Selective $k$-Coloring} parameterized by the vertex number and number of parts.

\noindent\textbf{Construction.}
Let $\mathcal{H} = (H_1, \dots, H_t)$ be the input instances to \pname{3-Coloring} on 4-regular graphs and $(G, \mathcal{V})$ the instance of \pname{Selective $k$-Coloring} we wish to build.
We may further suppose that all $t$ input graphs are simply sets of edges over the same ground set $V(\mathcal{H}) = [n]$.
For each vertex $v \in V(\mathcal{H})$, we add the vertex gadget $G_v$, which has vertex set $V(G_v) = \bigcup_{i \in [3]} A_i(v) \cup A_e(v) \cup Q(v)$; its edges are such that each $A_i(v) = \{v_i(u) \mid u \in V(\mathcal{H}) \setminus \{v\}\}$, $i \in [3] \cup \{e\}$, is an independent set, $(A_1(v), A_2(v), A_3(v))$ is a complete tripartite graph, and $Q(v)$ is a set of four disjoint cliques $Q_{1,2}(v), Q_{1,3}(v), Q_{2,3}(v), Q_e(v)$, each of size $k-1$, and no other edges exists.
We also partition $G_v$ into the family $\parts(G_v)$ so that, for each $u \in V(\mathcal{H}) \setminus \{v\}$, $\{v_1(u), v_2(u), v_3(u), v_e(u)\} \in \parts(G_v)$ and, for each $x \in Q(v)$, we also have $\{x\} \in \parts(G_v)$, which implies $|\parts(G_v)| = n + 4k - 5$.
An example with $n = 5$ is given in Figure~\ref{fig:comp_vertex_gadget}.

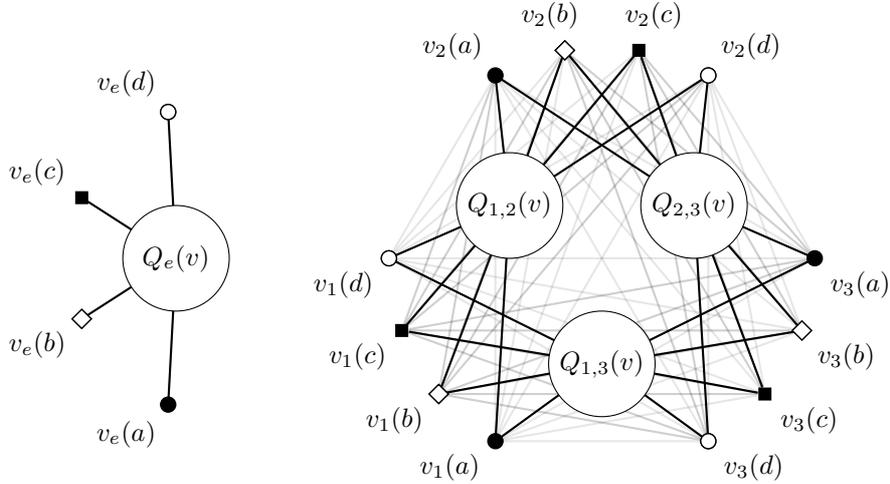
\begin{figure}[!htb]
    \centering
     \begin{tikzpicture}[scale=0.7]
        \GraphInit[unit=3,vstyle=Normal]
        \SetVertexNormal[Shape=circle, FillColor=black, MinSize=3pt]
        \tikzset{VertexStyle/.append style = {inner sep = \inners, outer sep = \outers}}
        \SetVertexLabelOut
        \begin{scope}[xshift=-5cm]
            \Vertex[a=247.5, d=3, Lpos=247.5, Math, L={v_e(a)}]{ae}
            \Vertex[a=202.5, d=3, Lpos=202.5, Math, L={v_e(b)}]{be}
            \Vertex[a=157.5, d=3, Lpos=157.5, Math, L={v_e(c)}]{ce}
            \Vertex[a=112.5, d=3, Lpos=112.5, Math, L={v_e(d)}]{de}
            \Vertex[x=-1, y=0, NoLabel]{qe}
            \Edges(ae,qe,be)
            \Edges(ce,qe,de)
            \draw[fill=white] (-1,0) circle(1cm);
            \node at (-1,0) {$Q_e(v)$};
        \end{scope}
        \begin{scope}[xshift=2cm]
            \begin{scope}
               \Vertex[a=240, d=4, Lpos=210, Math, L={v_1(a)}]{a1}
               \Vertex[a=220, d=4, Lpos=210, Math, L={v_1(b)}]{b1}
               \Vertex[a=200, d=4, Lpos=210, Math, L={v_1(c)}]{c1}
               \Vertex[a=180, d=4, Lpos=210, Math, L={v_1(d)}]{d1}
            \end{scope}
            
            \begin{scope}
               \Vertex[a=120, d=4, Lpos=120, Math, L={v_2(a)}]{a2}
               \Vertex[a=100, d=4, NoLabel]{b2}
               \node at (102:4.7) {$v_2(b)$};
               \Vertex[a=80, d=4, NoLabel]{c2}
               \node at (78:4.7) {$v_2(c)$};
               \Vertex[a=60, d=4, Lpos=60, Math, L={v_2(d)}]{d2}
            \end{scope}
            
            \begin{scope}
               \Vertex[a=360, d=4, Lpos=330, Math, L={v_3(a)}]{a3}
               \Vertex[a=340, d=4, Lpos=330, Math, L={v_3(b)}]{b3}
               \Vertex[a=320, d=4, Lpos=330, Math, L={v_3(c)}]{c3}
               \Vertex[a=300, d=4, Lpos=330, Math, L={v_3(d)}]{d3}
            \end{scope}
            
            \Vertex[a=150, d=2, NoLabel]{q12}
            \Vertex[a=270, d=2, NoLabel]{q13}
            \Vertex[a=30, d=2, NoLabel]{q23}
            \foreach \i in {a,b,c,d} {
               \foreach \j in {2,3} {
                   \Edge(q1\j)(\i1)
                   \Edge(q1\j)(\i\j)
                   \Edges[style={ opacity=0.1}](a1,\i\j,b1,\i\j,c1,\i\j,d1)
               }
               \foreach \j in {3} {
                   \Edge(q2\j)(\i2)
                   \Edge(q2\j)(\i\j)
                   \Edges[style={ opacity=0.1}](a2,\i\j,b2,\i\j,c2,\i\j,d2)
               }
            }
            \draw[fill=white] (150:2) circle(1cm);
            \node at (150:2) {$Q_{1,2}(v)$};
            \draw[fill=white] (270:2) circle(1cm);
            \node at (270:2) {$Q_{1,3}(v)$};
            \draw[fill=white] (30:2) circle(1cm);
            \node at (30:2) {$Q_{2,3}(v)$};
            
            \SetVertexNoLabel
            \begin{scope}
               \tikzset{VertexStyle/.append style = {shape = rectangle, inner sep = 2.2pt}}
               \Vertex[Node]{c1}
               \Vertex[Node]{c2}
               \Vertex[Node]{c3}
               \Vertex[Node]{ce}
            \end{scope}
            \begin{scope}
               \tikzset{VertexStyle/.append style = {shape = diamond, inner sep = 2pt}}
               \Vertex[Node]{b1}
               \Vertex[Node]{b2}
               \Vertex[Node]{b3}
               \Vertex[Node]{be}
               \AddVertexColor{white}{b1,b2,b3,be}
            \end{scope}
            \begin{scope}
               \tikzset{VertexStyle/.append style = {shape = circle, inner sep = 2pt}}
               \Vertex[Node]{a1}
               \Vertex[Node]{a2}
               \Vertex[Node]{a3}
               \Vertex[Node]{ae}
               \Vertex[Node]{d1}
               \Vertex[Node]{d2}
               \Vertex[Node]{d3}
               \Vertex[Node]{de}
               \AddVertexColor{white}{d1,d2,d3,de}
            \end{scope}
        \end{scope}
    \end{tikzpicture}
    \caption{Vertex gadget for vertex $v \in V(H)$ where $V(\mathcal{H}) = \{a,b,c,d,v\}$, and each $Q_{i,j}(v)$ is a clique of size $k-1$. Each differently shaped/colored set of four vertices corresponds to a different part of $\mathcal{V}$.}
    \label{fig:comp_vertex_gadget}
\end{figure}

To obtain $(G, \mathcal{V})$, we begin by adding to $G$ one copy of the vertex gadget $G_v$ for each $v \in V(\mathcal{H})$.
Now, for each $H_j \in \mathcal{H}$, we add a vertex $y_j$ to $G$, which we make adjacent to all vertices $u \in Q(G)$, where $Q(G) = \bigcup_{v \in V(H)} Q(v)$ and, for each pair $u,v \in V(\mathcal{H})$, if $u \in N_{H_j}(v)$, we add an edge between $y_j$ and $v_e(u)$, otherwise we add an edge between $y_j$ and $v_i(u)$, for each $i \in [3]$.
Finally, if there is some $H_j \in \mathcal{H}$ where $uv \in E(H_j)$, we add edge $v_i(u)u_i(v)$ to $G$, for every $i \in [3]$, and set $\mathcal{V} = \bigcup_{v \in V(\mathcal{H})} \parts(G_v) \cup \{Y\}$, where $Y =\{y_1, \dots, y_t\}$.
Intuitively, by choosing vertex $y_j \in Y$ to be part of the solution forces us to only consider the choices in $G \setminus Y$ that are compatible with the edges of $H_j$: we cannot pick the excess vertices of $A_e(v)$ corresponding to edges of $H$, otherwise $\{y_j, v_e(a)\} \cup Q_e(v)$ is a clique of size $k+1$, nor can we pick $v_i(b) \in A_i(v)$ corresponding to a non-edge of $E(H_j)$, otherwise $\{y_j, v_i(b)\} \cup Q_{i,\ell}(v)$ is a clique of size $k+1$.
Moreover, note that, since every $y_j \in Y$ is adjacent to all vertices of $Q(G)$ and each clique of $Q(G)$, the vertices of $G \setminus Q(G)$ we may pick \textit{must} be a single color class in any solution to $(G, \mathcal{V})$, i.e. they must be an independent set.
We formalize this as Observation~\ref{obs:ind_set}.

\begin{observation}
    \label{obs:ind_set}
    In any $k$-coloring of a subgraph $G^*$ of $G$ that contains some $y_j \in Y$ and $Q(G)$, $y_j$ and $u \in V(G^*) \setminus \{y\}$ have the same color if and only if $u \notin Q(G)$.
\end{observation}

\begin{lemma}
    \label{lem:composition}
    $(G,\mathcal{V})$ is a \YES\ instance of \pname{Selective $k$-coloring} if and only if some $H_j \in \mathcal{H}$ is 3-colorable.
\end{lemma}

\begin{proof}
    Let $\varphi$ be a 3-coloring of $H_j$.
    We pick a $k$-colorable subgraph $G^*$ of $G$ as follows: begin by adding $y_j$ to $G^*$; now, for each $v \in V(\mathcal{H})$, add every vertex in $Q(v)$ and in $A_e(v) \setminus N_G(y_j)$ to $G^*$.
    Finally, for each $i \in [3]$ and $v \in \varphi_i$, add $A_i(v) \setminus N_G(y_j)$ to $G^*$.
    To see that every $S \in \mathcal{V}$ has $S \cap V(G^*) \neq \emptyset$, note that, if $|S| = 1$, $S$ contains a vertex of some clique in $Q(G)$, which have been picked.
    Otherwise, either $S = Y$, which we have hit with $y_j$, or $S = \{v_1(u), v_2(u), v_3(u), v_e(u)\}$ for some $v \in V(\mathcal{H})$ and $u \in V(\mathcal{H}) \setminus \{v\}$.
    If $uv \notin E(H_j)$, then $v_e(u)$ has been picked, otherwise, since $v$ has been colored with a single color $i$, exactly one vertex of $S$ is in $V(G^*)$.
    We now color every vertex in $V(G^*) \setminus Q(G)$ with color $1$ and, for each $v \in V(\mathcal{H})$ and $Q \in Q(v)$, assign one of the remaining $k-1$ colors to a different vertex of $Q$.
    Since $Q(G)$ is a disjoint union of cliques and only color $1$ was used in vertices outside of $Q(G)$, if $I = V(G^*) \setminus Q(G)$ is an independent set we are done.
    So suppose to the contrary, that there is an edge $ab$ in $G^*[I]$.
    By the construction of $G$, this edge must be between vertices of different vertex gadgets, say $a \in G_u$ and $b \in G_v$ with $uv \in E(H_j)$, but it must be the case that $a = u_i(v)$ and $b = v_i(u)$, which contradicts the hypothesis that $\varphi$ is a 3-coloring of $H_j$.
    
    For the converse, let $G^*$ be a $k$-colorable subgraph of $G$ and $\psi$ one of its $k$-colorings.
    Note that $Q(G) \cup \{y_j\} \subseteq V(G^*)$ since, for each $x \in Q(G)$, $\{x\} \in \mathcal{V}$ and at least one $y_j$ must have been picked to hit $Y$.
    By Observation~\ref{obs:ind_set}, we conclude that $V(G^*) \setminus Q(G)$ is a color class of $\psi$.
    Moreover, $V(G^*) \cap G_v \setminus Q(v) = A_i(v) \setminus N_G(y_j)$, since every $S \in \parts(G_v)$ must be hit and the only vertices that can be picked are those not adjacent to $y_j$.
    To obtain the 3-coloring $\varphi$ of $H_j$, color $v \in V(H)$ with $i$ whenever there is some $v_i(u) \in V(G^*)$.
    Since $V(G^*) \setminus Q(G) \setminus \{y_j\}$ is an independent set, for every $uv \in E(H_j)$ we have that $\{v_i(u), u_i(v)\} \nsubseteq V(G^*)$, so $V(G^*) \cap A_i(u) = \emptyset$, which implies that $u$ was colored with $j \neq i$ and, consequently, $\varphi$ is a proper 3-coloring of $H_j$.
\end{proof}

\begin{lemma}
    Graph $G$ has a vertex cover with $\bigO{kn^2}$ vertices and $|\mathcal{V}| \in \bigO{n^2}$.
\end{lemma}

\begin{proof}
    For the first statement, note that $Y$ is an independent set and $G \setminus Y$ depends only on the size $n$ of the input instances: each $G_v$ has $|A_e(v)| + \sum_{i \in [3]} |A_i(v)| + |Q(v)| = 4(n + k - 2)$ vertices, implying $|V(G) \setminus Y| = 4n^2 + 4n(k - 2)$.
    For the latter, we again have $|\mathcal{V}| = 1 + \sum_{v \in V(\mathcal{H})} |\parts(G_v)| = n^2 + n(4k - 5) + 1$.
\end{proof}

\begin{theorem}
    \label{thm:no_kernel_vc}
    For every fixed $k \geq 1$, \pname{Selective $k$-Coloring} does not admit a polynomial kernel when jointly parameterized by the vertex cover number and the number of parts in the partition, unless $\NP \subseteq \coNP/\poly$.
\end{theorem}

Theorem~\ref{thm:no_kernel_vc} is a directly implied by the previous two lemmas.
It is important to observe that at no point in the proof of Lemma~\ref{lem:composition} we rely on the fact that $Y$ is an independent set, the only important feature is that $Y \in \mathcal{V}$; as such we may encode a member of whichever non-trivial graph class $\mathcal{G}$ we desire in $G[Y]$.

\begin{corollary}
    For every fixed $k \geq 1$ and non-trivial graph class $\mathcal{G}$, \pname{Selective $k$-Coloring} does not admit a polynomial kernel when parameterized by distance to $\mathcal{G}$ and the number of parts of $\mathcal{V}$, unless $\NP \subseteq \coNP/\poly$.
\end{corollary}

The most important consequence of Theorem~\ref{thm:no_kernel_vc}, however, is when we fix $k=1$, i.e. we are simply looking for an independent set containing one vertex of each part of $\mathcal{V}$.
This problem is the widely used \pname{Multicolored Independent Set} problem and, to the best of our knowledge, the negative kernelization result under vertex cover we present was not previously known.
This is in contrast with the even more classic  \pname{Independent Set} problem, which was known to admit a polynomial kernel when parameterized by vertex cover~\cite{minor_free_kernel}.

\begin{corollary}
    \pname{Multicolored Independent Set} does not admit a polynomial kernel when jointly parameterized by the vertex cover number and the number of colors, unless $\NP \subseteq \coNP/\poly$.
\end{corollary}

%% file: conclusion.tex
\section{Final Remarks}

We presented an initial study of the parameterized complexity of the \pname{Selective Coloring} problem by showing parameterizations that lead to fixed-parameter tractability, but also showing that they do not allow us to find a polynomial kernel for the problem unless $\NP \subseteq \coNP\poly$.
Specifically, we proved that it is fixed parameter tractable when parameterized by distance to cluster, by treewidth and number of parts, and by cotreewidth and number of parts, generalizing a result of Demange et al.~\cite{selective_complexity} and showing that others are, in a sense, optimal.
Our most interesting contribution, however, is the proof that for every fixed $k \geq 1$, \pname{Selective $k$-Coloring} does not admit a polynomial kernel when parameterized by vertex cover and number of parts unless $\NP \subseteq \coNP/\poly$, which implies that \pname{Multicolored Independent Set} has no kernel under the same parameterization and assumption.
As future work, we would like to determine how local properties on the parts of $\mathcal{V}$ or between the parts may aid in the design of parameterized algorithms for typically intractable parameterizations, specially for the parameterization max leaf number and number of parts, for which we were unable to design a better algorithm then the one given in Section~\ref{sec:tw}.
We are also interested in investigating the parameterized complexity of the worst case scenario problem \pname{Max Selective Coloring}, i.e. finding a graph that hits every part of $\mathcal{V}$ and has maximum chromatic number.
While some of our ideas may translate naturally, the algorithms for treewidth and cotreewidth seem to break down.
It may also be interesting to examine the problem from the parameterized approximation point of view, as approximation was one of the main objectives of previous work on the subject~\cite{selective_complexity,min_max_selective}.